\documentclass[runningheads,envcountsame]{llncs}
\usepackage[T1]{fontenc}
\usepackage{graphicx}
\usepackage[utf8]{inputenc}
\usepackage[T1]{fontenc} 
\usepackage{amsmath,amssymb}
\usepackage{hyperref}
\hypersetup{
  pdfusetitle,
  colorlinks,
  linkcolor={red!50!black},
  citecolor={blue!50!black},
  urlcolor={blue!80!black}
}
\usepackage{cleveref}
\usepackage[inline]{enumitem}
\usepackage{fourier}
\usepackage[textsize=small]{todonotes}
\setuptodonotes{inline}

\usepackage{thm-restate}

\setlist[enumerate]{label=(\roman*)}
\setlist{nosep}

\newcommand{\R}{\mathbb{R}}
\newcommand{\smallmat}[1]{\left[ \begin{smallmatrix} #1 \end{smallmatrix} \right]}
\newcommand{\Z}{\mathbb{Z}}
\newcommand{\N}{\mathbb{N}}

\newcommand{\zero}{\mathbf{0}}

\begin{document}

\title{Totally $\Delta$-modular IPs with two non-zeros in most rows}

\author{Stefan Kober\inst{1}\orcidID{0000-0003-2610-1494}}
\authorrunning{S. Kober}

\institute{\ \inst{1}Université Libre de Bruxelles}

\maketitle

\begin{abstract}
    Integer programs (IPs) on constraint matrices with bounded subdeterminants are conjectured to be solvable in polynomial time.
    We give a strongly polynomial time algorithm to solve IPs where the constraint matrix has bounded subdeterminants and at most two non-zeros per row after removing a constant number of rows and columns.
    This result extends the work by Fiorini, Joret, Weltge \& Yuditsky (J. ACM 72(1), 1-50 (2025)) by allowing for additional, unifying constraints and variables.
\end{abstract}
\section{Introduction}
In this work, we are concerned with integer programs (IPs) of the following form:
\begin{equation}
    \min \, \left\{ c^\intercal x : Mx \le b, \, x \in \Z^n \right\}, \tag{IP}\label{eq_ip}
\end{equation}
where $c\in\Z^n$, $M\in\Z^{m\times n}$, and $b\in\Z^m$ for $m,n\in\N$.
In general, assuming $\mathrm{P}\neq\mathrm{NP}$, one cannot hope for a polynomial-time algorithm to solve~\eqref{eq_ip}.
Still, there are certain subclasses of IPs that can be solved efficiently.
The most prominent example for this is IPs with \emph{totally unimodular} (TU) constraint matrices, i.e., if $M$ is TU.
Recall that a matrix $M\in\Z^{m\times n}$ is TU, if all its \emph{subdeterminants} are in $\{0,\pm1\}$, where a subdeterminant is the determinant of a square submatrix of $M$.

Totally unimodular matrices possess the nice geometric property that the convex hull of the feasible solutions to~\eqref{eq_ip} is equal to the space of feasible solutions of its linear relaxation $\min \, \left\{ c^\intercal x : Mx \le b, \, x \in \R^n \right\}$.
Therefore, we can solve such IPs by solving a single linear program. 
Such (combinatorial) linear programs can be solved in strongly polynomial time due to a result by Tardos~\cite{Tar86}.
This implies polynomial solvability for a collection of combinatorial optimization problems that can be stated as IPs on TU constraint matrices, such as the maximum-flow problem, the matching problem in bipartite graphs, or the stable set problem in bipartite graphs.

Almost 30 years ago, Shevchenko~\cite{She96} asked for generalizations of this result to IPs on constraint matrices with bounded subdeterminants.
We call a matrix $M$ \emph{totally $\Delta$-modular}, if all its subdeterminants are in $[-\Delta,\Delta]\cap\Z$.
A fundamental open question in combinatorial optimization is the following:

\emph{Can we solve \eqref{eq_ip} in polynomial time, when $M$ is totally $\Delta$-modular?}

In a breakthrough paper, Artmann, Weismantel \& Zenklusen~\cite{AWZ17} answered the question in the affirmative for $\Delta=2$.
Generalizing their ideas to $\Delta\ge3$ has proven to be difficult, despite some progress, see e.g. Nägele, Santiago \& Zenklusen~\cite{NSZ24}.
Fiorini, Joret, Weltge \& Yuditsky~\cite{FJWY25} fix $\Delta$ to be an arbitrary constant, and further require a structural condition on $M$, see \Cref{thm_twononzeros}.

\begin{theorem}[Fiorini, Joret, Weltge \& Yuditsky~\cite{FJWY25}]\label{thm_twononzeros}
    Let $\Delta\in\N$.
    There is a polynomial-time algorithm to solve~\eqref{eq_ip} if $M$ is totally $\Delta$-modular and has at most two non-zeros per row (or column).
\end{theorem}

Another example for restricting the structure of the constraint matrix are \emph{nearly totally unimodular matrices} as studied by Aprile et al.~\cite{AFJ24}.
The same idea has also proven helpful for a strongly polynomial-time algorithm for linear programming, see Dadush et al.~\cite{DKN24}, who also study the case of two non-zeros per row (or column).

We follow this second stream of research and investigate the question in the case where $\Delta$ is arbitrary but fixed, and there is a further structural condition on the constraint matrix.
In the context of this work, we study the following parametrized problem, and provide a strongly polynomial-time algorithm.
In particular, our result combines ideas from Fiorini et al.~\cite{FJWY25} and Aprile et al.~\cite{AFJ24}.
\begin{problem}\label{prob_general}
    Fix $\Delta,k\in\N$.
    Let $m_1,n_1\in\N$ and $m_2,n_2\le k$, and define $m:=m_1+m_2$, $n:=n_1+n_2$.
    Let $M=\smallmat{A&B\\W&D}\in\Z^{m\times n}$ be totally $\Delta$-modular, where $A\in\Z^{m_1\times n_1}$, $B\in\Z^{m_1\times n_2}$, $W\in\Z^{m_2\times n_1}$ and $D\in\Z^{m_2\times n_2}$. 
    Assume $A$ has at most two non-zeros per row.
    Let $c\in\Z^{n}$ and $b\in\Z^{m}$.
    Solve \eqref{eq_ip}.
\end{problem}

We reduce Problem~\ref{prob_general} to Problem~\ref{prob_reduced} below, in which the main technical difficulty for this project lies.

\begin{problem}\label{prob_reduced}
    Fix $\Delta,k\in\N$.
    Let $A\in\{0,\pm1\}^{m\times n}$, $W\in\Z^{k\times n}$, let $M=\smallmat{A\\W}$ be totally $\Delta$-modular and assume $A$ has two non-zeros per row.
    Let $c,\ell,u\in\Z^n$, $b\in\Z^m$ and $d\in\Z^k$.
    Solve $\ \min \, \left\{ c^\intercal x : Ax \le b, \, Wx=d, \, \ell\le x\le u,\, x \in \Z^n \right\}$.
\end{problem}

In the \emph{($k$-dimensional) partially ordered knapsack problem}, we are given an integer $k\in\N$, a set of elements $V$ with an underlying partial order, a profit function $p:V\to\Z$, $k$ weight-functions $W_i:V\to\Z$ for $i\in[k]$, and $k$ budgets $b_i\in\Z$ for $i\in[k]$.
The goal is to find the maximum profit subset of $V$ that respects the budgets with respect to the weight functions and adheres to the partial order as precedence constraints, i.e., in order to choose any element, we need to choose all its predecessors with respect to the partial order.
This problem is known to be $\mathrm{NP}$-complete even if there is only one weight function with bounded entries, and the partial order is bipartite~\cite{JN83}.
It can be formulated as an IP on a totally unimodular constraint matrix that has two non-zeros per row, and $k$ extra constraints.
In the following, we refer to this formulation as the respective IP.
This makes it an interesting problem to consider in the context of this work.
Note that the Hasse diagram of the partial order induces a directed graph $G$ on $V$.

\subsection{Contribution}
Our main result is a strongly polynomial-time algorithm for IPs as defined in Problem~\ref{prob_general}, see \Cref{thm_main}.
By combining our structural results with a proximity result by Aprile et al.~\cite{AFJ24}, our proof technique implies an FPT-algorithm for Problem~\ref{prob_reduced} if $A$ is totally unimodular, see \Cref{cor_tu} and in particular for the $k$-dimensional partially order knapsack problem, see \Cref{cor_pok}.
For this problem, only few non-trivial polynomially solvable cases are known~\cite{KS07}.
A proof for \Cref{cor_tu} has previously been given in the author's thesis~\cite{Kob23}.

Proofs for the statements are deferred to \Cref{sec:dp}.

\begin{theorem}\label{thm_main}
    There is a strongly polynomial-time algorithm to solve Problem~\ref{prob_general}.
\end{theorem}
\begin{corollary}\label{cor_tu}
    There is an FPT-algorithm to solve Problem~\ref{prob_reduced} with respect to the largest subdeterminant if $A$ is totally unimodular.
\end{corollary}
\begin{corollary}\label{cor_pok}
    There is an FPT-algorithm for the $k$-dimensional partially ordered knapsack problem with parameter 
    \[\Delta=\max\{\sum_{v\in V(H)}W_i(v):i\in[k],H\text{ connected subgraph of } G\}.\]
\end{corollary}

Note that for $k=1$, $\Delta$ in \Cref{cor_pok} corresponds to the largest subdeterminant of the corresponding constraint matrix, see~\cite[Corollary~4.24]{Kob23}.
For constant $k$, the subdeterminants of the respective IP are bounded by a function in $k$ and $\Delta$. 

Our results contribute to the research on IPs with bounded subdeterminants in the following ways.
We combine ideas from two previous technical papers~\cite{FJWY25,AFJ24} and push the boundary of IPs with bounded subdeterminants in terms of polynomial-time algorithms.
In particular, we relate the subdeterminants of matrices that have at most two non-zeros after removing one row to a novel invariant of weighted (signed) graphs, called the \emph{alternating weight of a tree}, see \Cref{prop_alt}.
We show that a bound on the alternating weight of a tree in an integer-weighted graph gives a structural restriction on the structure of non-zero weights within the graph.
More specifically, given an instance of Problem~\ref{prob_reduced}, we obtain a weighted graph and show that graphs corresponding to matrices with bounded subdeterminants admit a structured decomposition, see \Cref{thm_refinedstruct}.
We exploit this decomposition in a dynamic program that allows us to solve Problem~\ref{prob_general}, see \Cref{thm_main}.

We remark that it is unclear whether a constant proximity and decomposition result similar to the one in~\cite{AFJ24} holds in our case.
Therefore, in \Cref{sec:struct} we crucially need to exploit the fact that we work in a different graph class within the context of this work, and provide a substantially stronger decomposition.
In particular, Aprile et al.~\cite{AFJ24} show that graphs associated with their instances admit a decomposition such that the bags of the tree-decomposition either contain a \emph{small} number of vertices, correspond to a leaf of the tree and contain a \emph{small} number of roots, or are \emph{almost embeddable in a surface of bounded genus}.
We can show that the latter type is not necessary in graphs occuring in the context of this work.
Finally, our structural results are largely self-contained and in particular do not make use of the graph minor structure theorem for $K_t$-minor free graphs as required in~\cite{FJWY25,AFJ24}.

\subsection{Background and related work}
The main focus of our work lies in integer programming with bounded subdeterminants, see \Cref{fig_cont} for an overview of recent results.
There has been a large number of exciting results in the last years that consider the question of polynomial solvability, e.g., the aforementioned papers~\cite{AWZ17,FJWY25,NSZ24,AFJ24}, if the matrix is $\{a,b,c\}$-modular~\cite{GSW22}, if the problem corresponds to the total matching problem~\cite{FFKY24}, or for restricted subdeterminants over a polynomial ring~\cite{CKW24}.
Further, there is a randomized polynomial-time algorithm for feasibility of IPs on constraint matrices that are strictly $\Delta$-modular for $\Delta\in\{3,4\}$~\cite{NSZ24,NNSZ24}.

\begin{figure}[ht]
    \centering
    \includegraphics[width=0.65\textwidth]{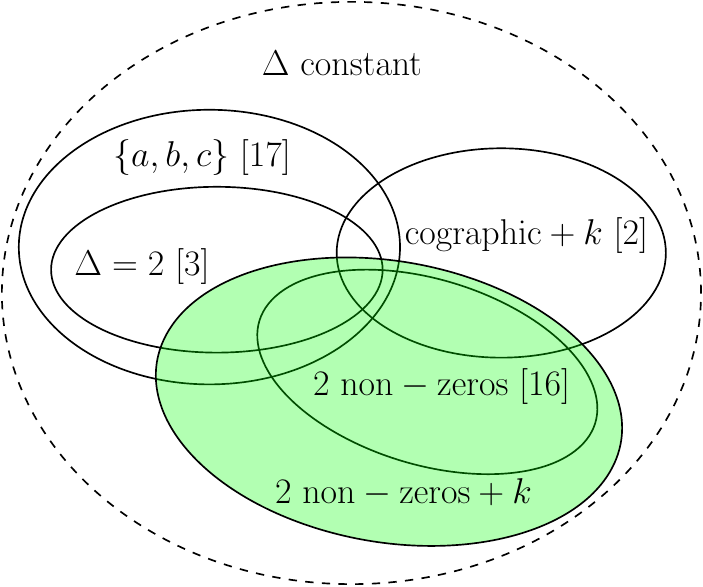}
    \caption{Recent results on IPs with bounded subdeterminants: bimodular matrices~\cite{AWZ17}, $\{a,b,c\}$-modular matrices~\cite{GSW22}, two non-zeros per row ($2$ non-zeros)~\cite{FJWY25}, nearly totally unimodular matrices, the cographic case (cographic $+k$)~\cite{AFJ24}, and this work ($2$ non-zeros $+k$, shaded in green).}
    \label{fig_cont}
\end{figure}

Secondly, we work on IPs over a class of constraint matrices \emph{with a constant number of extra rows}.
Such problems appear naturally within combinatorial optimization or operations research, where the extra rows can play the role of budget or exactness constraints.
Since already the knapsack problem is $\mathrm{NP}$-hard, which only has one budget-constraint, one typically needs bounds on the entries or subdeterminants to design polynomial-time algorithms for such problems.
We specifically mention the exact matching problem~\cite{PY82}, integer programming on few constraints~\cite{EW19,AHO20}, the exact matroid basis problem~\cite{LX24,ERW24}, integer programming over TU matrices with few extra constraints~\cite{AFJ24}, or the partially ordered knapsack problem~\cite{KS07}.

One important problem in the class of IPs with bounded subdeterminants and two non-zeros per row is the maximum stable set problem with bounded \emph{odd cycle packing number}, i.e., there is a bound on the maximum number of disjoint odd cycles. 
This problem has an efficient algorithm due to Fiorini et al.~\cite{FJWY25}, which we use in order to solve our base case.
Two exemplary problems that are covered by Problem~\ref{prob_general} are the budgeted stable set problem studied by Dhahan~\cite{Dha20}, and the partially ordered knapsack problem, where we further restrict to bounded subdeterminants.

\subsection{Overview}
We begin by setting up definitions that are used throughout the paper in \Cref{sec:prelim}.
In \Cref{sec:reduction}, we reduce Problem~\ref{prob_general} to Problem~\ref{prob_reduced}.
We proceed to exploit the bound on the subdeterminants in order to obtain bounds on the parameters of a related graph, as well as a structural description of such graphs in \Cref{sec:subdet,sec:struct}.
We conclude the proof of our main result in \Cref{sec:dp} by providing a dynamic program.
In the final \Cref{sec:conclusion}, we mention related problems and directions for further research.

\section{Preliminaries}\label{sec:prelim}
Throughout the paper, we use the notation $[k]:=\{1,\dots,k\}$.
Given an integer matrix $A\in\Z^{m\times n}$, we denote its largest subdeterminant by 
\[\Delta(A):=\max\{|\det(B)|:B\text{ is a submatrix of }A\}.\]
We refer to Diestel~\cite{Die10} for more background on the notions from graph theory that we use.
We proceed to give an overview of the most important tools and definitions we use in the remainder of the paper.

Given a graph $G=(V,E)$, and a subset of vertices $U\subseteq V$, we define the \emph{neighborhood} $N(U)\subseteq V-U$ of $U$ in $G$, such that for every $v\in N(U)$, there is an edge $vv'\in E$ with $v'\in U$.
Given an edge $e=vv'$ of $G$, we denote by $G/e$ the graph obtained by \emph{contracting} an edge $e$, where $v$ and $v'$ are replaced by one vertex $v_e$ that is adjacent to all vertices in $N(\{v,v'\})$.
We say for graphs $G$ and $H$ that $G$ contains an \emph{$H$-minor}, if we can obtain a graph isomorphic to $H$ by deleting vertices, deleting edges, and contracting edges of $G$.
We call a graph $G$ a \emph{subdivision} of a graph $H$, if we can obtain a graph isomorphic to $G$ by replacing some of the edges of $H$ by paths.

Given a graph $G$, we define a {\em tree-decomposition} of $G$ as a tuple $(T,\mathcal{B})$, such that $T$ is a tree, and a collection of bags $\mathcal{B}=(B_t)_{t\in V(T)}$ where $B_t\subseteq V(G)$ for each $t\in V(T)$ is a subset of the vertices of $G$ such that the following properties are satisfied: 
\begin{itemize}
    \item $\bigcup_{t\in V(T)}B_t=V(G)$,
    \item for each edge $vv'\in E(G)$, there is a vertex $t\in V(T)$ such that $\{v,v'\}\subseteq B_t$, and
    \item for each $v\in V(G)$, the set $\{t\in V(T):v\in B_t\}$ induces a connected graph in $T$.
\end{itemize}
If $T$ is a star, we also call $(T,\mathcal{B})$ a {\em star-decomposition}.
We define the \emph{adhesion-size} of a tree-decomposition as $\ \max\{|B_t\cap B_{t'}|:tt'\in E(T)\}$.
Note that we allow $G$ to not be connected, and $B_t$ to be empty for some $t\in V(T)$. 

\section{Reduction}\label{sec:reduction}
\begin{theorem}[Cook, Gerards, Schrijver, \& Tardos~\cite{CGST86}]\label{thm_cookproximity}
    There exists a function $f_{\ref{thm_cookproximity}}:\Z^2\to\Z$ with the following properties.
    Fix $\Delta\in\N$ and consider an integer program as in~\eqref{eq_ip}.
    If $M$ is totally $\Delta$-modular, \eqref{eq_ip} is feasible, and $x^*$ is an optimal solution of the linear relaxation, then there exists an optimal solution $z^*$ of~\eqref{eq_ip} with $\|x^* - z^*\|_\infty \le f_{\ref{thm_cookproximity}}(n,\Delta)$.
\end{theorem}
Cook, Gerards, Schrijver, \& Tardos~\cite{CGST86} show that $f_{\ref{thm_cookproximity}}(n,\Delta)\le n\Delta$.
This was recently improved by Celaya, Kuhlmann, Paat \& Weismantel~\cite{CKPW23} to $f_{\ref{thm_cookproximity}}(n,\Delta)<\frac{4n+2}{9}\Delta$.
If the constraint matrix has a special structure, it is possible to give a bound not depending on the dimension of the problem, see e.g. Aliev, Henk \& Oertel~\cite{AHO20}, or Aprile et al.~\cite{AFJ24}.
In fact, it is not clear that the dimension should appear in the proximity function at all.

The following statement relates the two problems from the introduction.
Clearly, any instance of Problem~\ref{prob_reduced} is also an instance of Problem~\ref{prob_general}.
But in fact, there is a polynomial-time reduction from Problem~\ref{prob_general} to Problem~\ref{prob_reduced}.
Note that we need to solve a polynomial number of instances of Problem~\ref{prob_reduced} in order to solve Problem~\ref{prob_general}.

\begin{proposition}\label{prop_reduction}
    If there is a polynomial-time algorithm to solve Problem~\ref{prob_reduced}, then there is a polynomial-time algorithm for Problem~\ref{prob_general}.
\end{proposition}

The proof of the statement is based on the proximity result by Cook et al.~\cite{CGST86}, see \Cref{thm_cookproximity}.
This theorem allows us to guess on a constant number of variables in any totally $\Delta$-modular instance of \eqref{eq_ip}.
Fiorini et al.~\cite[Lemma~4]{FJWY25} show that $A$ has at most $\log_2(\Delta)$ columns with entries larger than $1$ in absolute value.
Thus, we can efficiently enumerate all possible assignments to the corresponding variables.
Therefore, we can solve a polynomial number of instances of Problem~\ref{prob_reduced} in order to solve an instance of Problem~\ref{prob_general}.

In addition, we recenter the problem around $\zero_n$, i.e., given a solution $x^*\in\R^n$ of the linear relaxation, we redefine our variables $x':=x-\lfloor x^*\rfloor$.
By \Cref{thm_cookproximity} there is an optimal solution $z^*$ to the integer program, such that $\|z^*\|_\infty\le f_{\ref{thm_cookproximity}}(n,\Delta)$.
Thus, we can assume without loss of generality that the lower and upper bounds have bounded entries, i.e., $\ell,u\in[-f_{\ref{thm_cookproximity}}(n,\Delta),f_{\ref{thm_cookproximity}}(n,\Delta)]^n\cap\Z^n$.

For a detailed description of the reduction, we refer to Fiorini et al.~\cite{FJWY25}, see Lemma~4 and the following discussion.

\section{Subdeterminants}\label{sec:subdet}
Let $M=\smallmat{A\\w^\intercal}$, where $A\in\{0,\pm1\}^{m\times n}$ with at most two non-zeros per row, and $w\in\Z^n$.
In this section, we show a relationship between the subdeterminants of $M$, and the structure of a corresponding graph.

Let $G=(V,E)$ be a graph, and $\sigma:E\to\{-1,+1\}$ be a signing of its edges.
We call the triple $S=(V,E,\sigma)$ a \emph{signed graph}.
We call a subgraph $S'$ of $S$ \emph{odd} if $\prod_{e\in E(S')}\sigma(e)=-1$, and \emph{even} otherwise.
Observe that if $\sigma(e)=-1$ for all $e\in E$, then the notion of an \emph{odd cycle} coincides with the notion of an odd cycle for unsigned graphs.
The \emph{odd cycle packing number} of a signed graph $S$ is denoted by $\mathrm{ocp}(S)$, and is defined as the maximum number of vertex-disjoint odd cycles.
For a signed graph $S$, we define an edge-vertex incidence matrix $M\in\{0,\pm1\}^{E\times V}$ in the following way.
Given an edge $e=vv'$, $M_{e,v}M_{e,v'}=-\sigma(e)$, i.e., if $e$ is odd, then $M_{e,v}=+1$ and $M_{e,v'}=+1$, or $M_{e,v}=-1$ and $M_{e,v'}=-1$, and if $e$ is even, then $M_{e,v}=+1$ and $M_{e,v'}=-1$, or $M_{e,v}=-1$ and $M_{e,v'}=+1$.
For $v''$ not incident to $e$, we set $M_{e,v''}=0$.
Note that given a signed graph, its edge-vertex incidence matrix is not necessarily unique, but any $\{0,\pm1\}$-matrix with two non-zeros per row is the edge-vertex incidence matrix of a signed graph, and thus has a unique associated signed graph.

The following result is well-known.
\begin{proposition}[Grossman, Kulkarni, \& Schochetman~\cite{GKS95}]\label{prop_ocp}
    Let $A\in\{0,\pm1\}^{m\times n}$ be the edge-vertex incidence matrix of a signed graph $S$.
    Then $2^{\mathrm{ocp(S)}}=\Delta\left(A\right)$.
\end{proposition}

Let $S$ be a signed graph, denote the set of subgraphs of $S$ that are trees by $\mathcal{T}(S)$ and let $T\in\mathcal{T}(S)$.
For $v,v'$ vertices of $T$, we denote the edges of the unique $v-v'$-path in $T$ by $P_T(v,v')$.
Note that $T$ induces a partition on $V(T)$ as follows.
We define $X_T,Y_T\subseteq V(T)$, such that $X_T\cap Y_T=\emptyset$, $X_T\cup Y_T=V(T)$, any even edge of $T$ is in $X_T$ or $Y_T$, and any odd edge of $T$ crosses between $X_T$ and $Y_T$.
Hence, $P_T(v,v')$ is even if and only if $v$ and $v'$ are both in $X_T$ or in $Y_T$.
This partition can be found greedily.

\begin{definition}
    Let $S=(V,E,\sigma)$ be a signed graph and $w:V\to\Z$ be a weight function on the vertices of $S$.
    Let $T\in\mathcal{T}(S)$.
    We define the \emph{alternating weight of a tree} as 
    \[
        \mathrm{alt}(T,w):=\left|\sum_{v\in X_T}w(v)-\sum_{v\in Y_T}w(v)\right|.
    \] 
    We define $\mathrm{alt}(S,w):=\max_{T\in\mathcal{T}(S)}\mathrm{alt}(T,w)$ to be the maximum alternating weight of a tree in $S$.
\end{definition}
Note that $\mathrm{alt}(T,w)$, and $\mathrm{alt}(S,w)$ are well-defined, since the bipartition is unique up to exchanging the sides, and we take the absolute value.
The alternating weight of a tree appears naturally in the context of subdeterminants.

\begin{proposition}\label{prop_alt}
    Let $A\in\{0,\pm1\}^{m\times n}$ be the edge-vertex incidence matrix of a signed graph $S=(V,E,\sigma)$ and $w:V\to\Z$ be a weight function on the vertices of $S$.
    Then $\mathrm{alt}(S,w)\le\Delta\left(\smallmat{A\\w^\intercal}\right)$.
\end{proposition}
\begin{proof}
    We show that any tree $T$ of $S$ gives rise to a subdeterminant of absolute value $\mathrm{alt}(T,w)$ within the matrix $\smallmat{A\\w^\intercal}$.
    We prove this by induction on the number of vertices of $T$.
    Given a tree $T$, we denote by $M$ the square submatrix of $\smallmat{A\\w^\intercal}$, where the rows are indexed by $E(T)$ and $w$, and the columns are indexed by $V(T)$.

    If $T$ has just a single vertex $v_1$ then $\mathrm{alt}(T,w)=|w(v_1)|=|\det(M)|$.
    Assume that $T$ has $t\ge2$ vertices.
    Then $T$ has a leaf $v_1$ with an edge $e=v_1v_2$ to some vertex $v_2\in V(T)$.

    Recall for the following that for an edge $e=(v,v')$, $\sigma(e)=-M_{e,v}M_{e,v'}$.
    In $M$, we subtract the row indexed by $e$ scaled with a factor of $w(v_1)\cdot M_{e,v_1}$ from the row indexed by $w$ to obtain $M'$.
    To be precise, we define $M'_{w,v_1}:=0$, $M'_{w,v_2}:=w(v_2)+\sigma(e)\cdot w(v_1)$, and $M'_{i,j}:=M_{i,j}$ otherwise, for $i\in E(T)\cup \{w\}$, and $j\in V(T)$.
    Since $M'$ is obtained from an elementary row operation on $M$, we observe that $\det(M')=\det(M)$.
    In $M'$, the column indexed by $v_1$ has exactly one non-zero entry of absolute value one, so we can remove the row indexed by $e$ and the column indexed by $v_1$ to obtain a matrix of size $(t-1)\times(t-1)$ while preserving the determinant of $M'$ in absolute value.

    We define $T':=T-\{v_1\}$, and $w':V(T')\to\Z$, where $w'(v_2):=w(v_2)+\sigma(e)\cdot w(v_1)$, and $w'(v):=w(v)$ for $v\in V(T')-\{v_2\}$.
    By the induction hypothesis, it follows that
    \[
        |\det(M)|=|\det(M')|=\left|\sum_{v\in A_{T'}}w'(v)-\sum_{v\in B_{T'}}w'(v)\right|=\left|\sum_{v\in X_T}w(v)-\sum_{v\in Y_T}w(v)\right|.
    \]
    \qed
\end{proof}

We remark that \Cref{prop_alt} also follows from the \emph{max-circuit imbalance} of an associated matrix, see~\cite[Proposition~3.2]{ENV22} for a relation between subdeterminants and the imbalance of circuits.
In contrast, we are not aware of a proof of \Cref{prop_ocp} based solely on the structure of circuits of the matrix.

\section{Structure}\label{sec:struct}
Given a graph $G=(V,E)$ and $k$ weight-functions $w_1,\dots,w_k:V\to\Z$, we define an associated set of terminals $R:=\{v\in V\mid\exists i\in[k]:w_i(v)\neq0\}$ as the set of vertices whose weight is non-zero with respect to at least one of the weight functions.
For an integer $r\in\N$, we call $G$ \emph{$r$-decorated} if $G$ is a tree, and at least $r$ of its leaves are terminals.

\begin{lemma}\label{lemma_decorated}
    Let $S=(V,E,\sigma)$ be a signed graph, $k,r\in\N$, and let $w_1,\dots,w_k:V\to\Z$ be $k$ weight-functions on the vertices.
    Let $G:=(V,E)$, and $R:=\{v\in V\mid\exists i\in[k]:w_i(v)\neq0\}$.
    If $G$ has an $r$-decorated tree as a subgraph, then there is some $i\in[k]$ such that $\mathrm{alt}(S,w_i)\ge\frac{r}{2k}$.
\end{lemma}
\begin{proof}
    Let $T$ be an $r$-decorated tree in $G$ with respect to $R$.
    Let $L(T)$ denote the set of leaves of $T$, and define the set of leaves which are terminals as $L':=L(T)\cap R$.
    Note that $|L'|\ge r$.
    By the pigeon-hole principle, there is some $i\in[k]$ such that $w_i(v)\neq0$ for at least $\frac rk$ vertices in $L'$.

    Fix a bipartition $X_T,Y_T$ of $T$ in order to compute the alternating weight of $T$.
    We denote the set of leaves with a positive contribution to $\mathrm{alt}(T,w_i)$ by $L^+:=\{v\in L\cap X_T\mid w_i(v)>0\}\cup\{v\in L\cap Y_T\mid w_i(v)<0\}$, and the set of leaves with a negative contribution to $\mathrm{alt}(T,w_i)$ by $L^-:=\{v\in L\cap X_T\mid w_i(v)<0\}\cup\{v\in L\cap Y_T\mid w_i(v)>0\}$.
    Clearly, $T^+:=T\setminus L^-$ and $T^-:=T\setminus L^+$ are trees, and $|\mathrm{alt}(T^+,w_i)-\mathrm{alt}(T^-,w_i)|\ge\frac rk$.
    Thus, 
    \[
        \mathrm{alt}(S,w_i)\ge\max\{|\mathrm{alt}(T^+,w_i)|,|\mathrm{alt}(T^-,w_i)|\}\ge\frac{r}{2k}.
    \]
    \qed
\end{proof}

\begin{theorem}[{Ding~\cite[Lemma~3.1]{Din17}}]\label{thm_struct}
    Let $G=(V,E)$ be a graph with some set of terminals $R\subseteq V$, and $r\in\N$.
    There exists a function $f_{\ref{thm_struct}}(r)\in O(r^2)$ with the following properties.
    If $G$ has no $r$-decorated tree as a subgraph, then there is a star-decomposition $(T,\mathcal{B})$ of $G$ (with $t_0$ the center of $T$), such that
    \begin{itemize}
        \item $|B_{t_0}\cap R|\le f_{\ref{thm_struct}}(r)$,
        \item $|B_{t_0}\cap B_t|=2$ for all $t\in V(T)\setminus\{t_0\}$, and both vertices in $B_{t_0}\cap B_t$ are in $R$,
        \item $B_t\cap B_{t'}=\emptyset$ for all $t,t'\in V(T)$ leaves of $T$,
        \item for all $t\in V(T)\setminus\{t_0\}$ it holds that $G[B_t]$ is connected, and each $r\in (B_t\setminus B_{t_0})\cap R$ is a cut vertex in $G[B_t]$ separating the two vertices in $B_{t_0}\cap B_t$.
    \end{itemize}
\end{theorem}

Note that this implies that the maximum degree of $T$ is bounded from above by $f_{\ref{thm_struct}}(r)/2$.
We further remark that all steps in the proof of~\cite[Lemma~3.1]{Din17} are constructive, so one can find the above decomposition in polynomial time.

In the proof of \cite[Lemma~3.1]{Din17}, the concrete order of $f_{\ref{thm_struct}}(r)$ is shown to be in $O(s(r)^2)$, where $s(r)$ denotes the minimal number such that every connected graph excluding a $K_{1,r}$-minor is a subdivision of a graph on $s(r)$ vertices.

%

\begin{lemma}\label{lemma_K1r}
    Let $r\in\N$, and $G$ be a connected graph without $K_{1,r}$-minor.
    There exists a function $s_{\ref{lemma_K1r}}(r)\in O(r)$, such that $G$ is a subdivision of a graph on at most $s_{\ref{lemma_K1r}}(r)$ vertices.
\end{lemma}
\begin{proof}
    Fix $r\in\N$, and let $G$ be a connected graph without $K_{1,r}$-minor.
    Then $G$ does not contain a tree with more than $r$ leaves as a subgraph.

    We first claim that given a tree $T$, the number of leaves is given by 
    \[
        \#(\text{leaves of }T)=2+\sum_{v\in V(T),\,\mathrm{deg}_T(v)\ge3}(\mathrm{deg}(v)-2).
    \]
    From the handshake lemma and the fact that in a tree, $|E(T)|=|V(T)|-1$, we know that $2+\sum_{v\in V(T)}(\mathrm{deg}(v)-2)=0$.
    By splitting the sum into vertices of degree less than $2$ and more than $2$, we obtain that 
    \[
        2+\sum_{v\in V(T),\,\mathrm{deg}_T(v)\ge3}(\mathrm{deg}(v)-2)=\sum_{v\in V(T),\,\mathrm{deg}_T(v)=1}(\mathrm{deg}(v)),
    \]
    which implies the claim.

    Let $p\in\N$, and consider a subgraph of $G$ that is a vertex-disjoint packing of $K_{1,r_i}$ with $r_i\ge3$ for all $i\in[p]$.
    We claim that 
    \begin{equation}\sum_{i\in[p]}r_i\le3r,\label{num_vertices}\end{equation} 
    so in particular $p\le r$.
    To see this, assume to the contrary that there is a packing of $K_{1,r_i}$ such that $r_i\ge3$ for all $i\in[p]$ and $\sum_{i\in[p]}r_i>3r$.
    Extend the packing to a spanning tree $T$ in $G$.
    By the formula for the number of leaves of a tree, we have 
    \begin{equation*}
        \#(\text{leaves of }T) = 2+\sum_{v\in V(T),\,\mathrm{deg}_T(v)\ge3}(\mathrm{deg}(v)-2)\ge 2+\sum_{i\in[p]}(r_i-2)>r,
    \end{equation*}
    a contradiction.
    The last inequality holds, since each $r_i\ge3$, and $\sum_{i\in[p]}r_i>3r$.

    Choose an inclusion-wise maximal set of vertices $N_1 = \{v_1,\dots,v_p\}$ of degree at least $3$ in $G$, such that $d(v_i,v_j)\ge 3$ for $i,j\in[p]$ with $i\neq j$.
    By~\eqref{num_vertices}, $|N_1|\le r$, and $\sum_{i\in[p]}\mathrm{deg}(v_i)\le 3r$.
    We define $N_2:=N(N_1)$, and observe that $|N_2|\le 3r$.
    We further define $N_3:=N(N_2)\setminus N_1$.
    Assume that $|N_3|>9r$.
    We construct a subset $N^*\subseteq N_2$ greedily, by iteratively choosing the vertex with the largest number of neighbors in $N_3$, adding it to $N^*$ and removing these neighbors from $N_3$.
    We stop the process when $N(N^*)\cap N_3>3r$.
    By the pigeonhole principle, each vertex $v\in N^*$ has at least $3$ unique neighbors in $N_3$ ( similarly, the first one has at least $4$ neighbors).
    Thus, $|N^*|\le r$, a contradiction to~\eqref{num_vertices}.
    Since $N_1$ is maximal, all vertices of degree at least $3$ must be contained in $N_1\cup N_2\cup N_3$.
    Further, there is at most $r$ vertices of degree $1$ in $G$ (otherwise every spanning tree has more than $r$ leaves).

    Thus, the number of vertices in $G$ with degree not equal to $2$ is in $O(r)$, which concludes the proof.
    \qed
\end{proof}

We proceed to modify the decomposition from \Cref{thm_struct} in order to obtain a graph decomposition that allows us to employ a dynamic programming framework in order to solve Problem~\ref{prob_reduced}.

\begin{lemma}\label{lem_helper}
    Let $G$ be a connected graph with set of terminals $R = \{v_0,\dots,v_\alpha\}$ for some $\alpha\in\N$.
    Assume that for all $i\in[\alpha-1]$, $v_i$ is a cut vertex of $G$ separating $v_0$ from $v_\alpha$.
    
    Then there is an ordering of $R$, such that every $v_0-v_\alpha$-path in $G$ visits all terminals in sequence with respect to the ordering.
    Further, any component of $G-R$ attaches to at most two terminals.
    These terminals must be consecutive with respect to the ordering.
\end{lemma}
\begin{proof}
    We first claim that any $v_0-v_\alpha$-path contains each terminal in $R$.
    Assume otherwise, then there is some $i\in[\alpha-1]$, such that $v_i$ does not separate $v_0$ from $v_\alpha$, a contradiction.

    If there are two paths with a different ordering of terminals, then there is some $v_i\in R$, such that its position is different in the orderings.
    Consider the connected components of $G-v_i$, and denote the components containing $v_0$ and $v_\alpha$ by $C_0$ and $C_\alpha$ respectively.
    By the first claim, every $v\in R-v_i$ is in $C_0$ or $C_\alpha$.
    Since $v_i$ is a cut vertex, every $v_0-v_\alpha$-path first traverses $C_0$ and then $C_\alpha$, a contradiction to the assumption that the position of $v_i$ is different.

    Therefore, we can assume that every $v_0-v_\alpha$-path visits all terminals in the same ordering.
    Now assume for a contradiction that there is a component $C$ of $G-R$ that attaches to $v_i$ and $v_j$ for $i<j\in\{0,\dots,\alpha\}$ where $i+1\neq j$.
    Then, we can find a $v_0-v_\alpha$-path avoiding $v_{i+1}$ by combining a $v_0-v_i$-path, a $v_i-v_j$-path in $C$, and a $v_j-v_\alpha$-path.
    This contradicts the first claim.
    \qed
\end{proof}

\begin{figure}[h]
    \centering
    \includegraphics[width=0.65\textwidth]{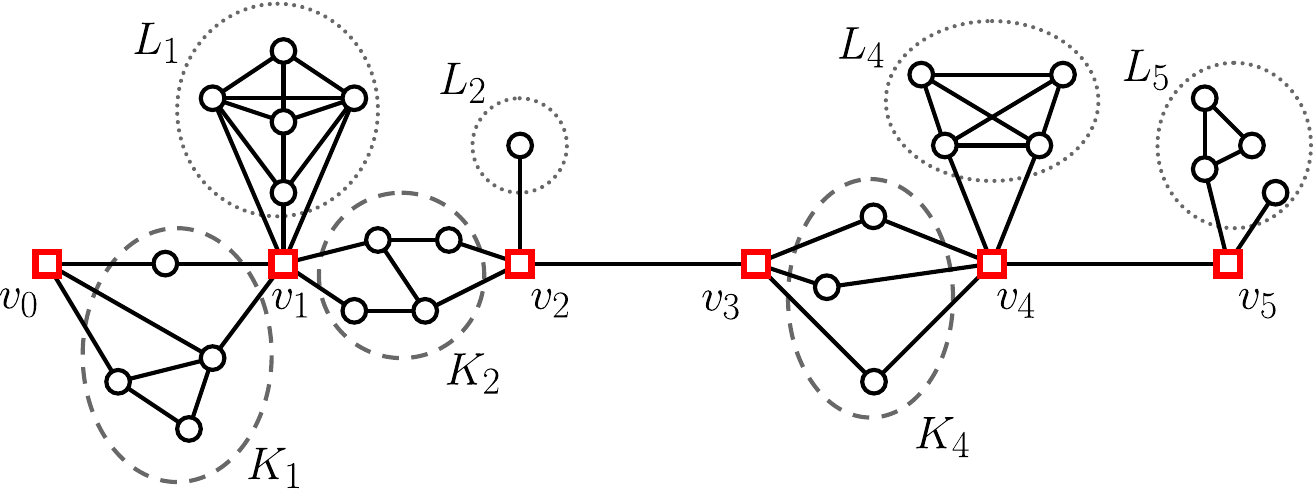}
    \caption{A depiction of the structure of graphs in \Cref{lem_helper}. The red vertices are terminals. The collections of components $K_i$ (dashed) and $L_j$ (dotted) are used in the proof of \Cref{thm_refinedstruct}.}
    \label{fig_path}
\end{figure}

\begin{theorem}\label{thm_refinedstruct}
    Let $r\in\N$ and $G$ be a graph with set of terminals $R$.
    There is a function $f_{\ref{thm_refinedstruct}}(r)\in O(r^2)$ with the following properties.
    If $G$ has no $r$-decorated tree as a subgraph, then there is a tree-decomposition $(T,\mathcal{B})$ of $G$ with adhesion-size and degree bounded by $f_{\ref{thm_refinedstruct}}(r)$, such that each bag $B_t$ for $t\in V(T)$ is of at least one of the following two types:
    \begin{enumerate}
        \item\label{type1} $|B_t|\le f_{\ref{thm_refinedstruct}}(r)$
        \item\label{type2} $t$ is a leaf of $T$ with neighbor $t'$, and $B_t\cap R\subseteq B_{t'}$
    \end{enumerate}
    Further, we can compute such a decomposition in polynomial time.
\end{theorem}

\begin{proof}
    Let $f_{\ref{thm_refinedstruct}}:=\max\{3,f_{\ref{thm_struct}}(r)/2+2,f_{\ref{thm_struct}}(r)\}$.
    We first assume that $G$ is a connected graph.
    Then, we can apply \Cref{thm_struct} and obtain a decomposition $(T',\mathcal{B}')$ of $G$.
    We proceed to describe how to modify such a decomposition in order to obtain a decomposition $(T,\mathcal{B})$ that fulfills the conditions of the lemma.
    
    For $k\ge1$, we define a {\em $k$-comb} as a tree on $2k$ vertices $c_1,\dots,c_k,d_1,\dots,d_k$, consisting of a path $c_1,\dots,c_k$, and edges $c_id_i$ for $i\in[k]$.
    We first describe how we modify $T'$ in order to obtain $T$.
    Let $w_0$ be the center vertex of $T'$.
    For any leaf $w$ of $T'$, we define $\lambda_w:=|B'_{w}\cap R|-1$.
    We replace each $w$ by a $\lambda_w$-comb on the vertex set $\{c^w_1,\dots,c^w_{\lambda_w},d^w_1,\dots,d^w_{\lambda_w}\}$, such that there is an edge $w_0c^w_1$.
    Finally, we attach a vertex $t_{-1}$ to the center vertex $w_0$, and denote the resulting tree by $T$ with center vertex $t_0$.

    We set $B_{t_{-1}}:=B'_{w_0}$ and $B_{t_0}:=B'_{w_0}\cap R$.
    Observe that $B_{t_0}$ has type~\ref{type1}, $B_{t_{-1}}$ has type~\ref{type2}, and their intersection is bounded by $f_{\ref{thm_struct}}(r)$.

    For any leaf $w$ of $T'$, we denote the two terminals in $B'_w\cap B'_{w_0}$ by $v_0(w)$ and $v_{\lambda_w}(w)$.
    By \Cref{lem_helper} applied to $B'_w$, we obtain an ordering $v_0(w),v_1(w),\dots,v_{\lambda_w}(w)$ of the terminals, such that any $v_0(w)-v_{\lambda_w}(w)$-path visits the terminals in this ordering.    
    We set $B_{c^w_i}:=\{v_{i-1}(w),v_i(w),v_{\lambda_w}(w)\}$ for $i\in[\lambda_w]$.
    Observe that each such bag has at most $3$ elements and is of type~\ref{type1}.

    It remains to describe the bags $B_{d^w_i}$ for $i\in[\lambda_w]$.
    By \Cref{lem_helper}, each component of $G[B'_w]\setminus R$ attaches to at most two consecutive terminals in $v_0(w),v_1(w),\dots,v_{\lambda_w}(w)$.
    For $i\in[\lambda_w]$, we denote by $K_i$ the union of all components that attach to $v_{i-1}(w)$ and $v_i(w)$.
    For $j\in\{0,\dots,\lambda_w\}$, we denote by $L_j$ the union of all components that attach only to $v_j(w)$.
    We define $B_{d^w_1}:=\{v_0(w),v_1(w),v_{\lambda_w}(w)\}\cup K_1\cup L_0\cup L_1$ and $B_{d^w_i}:=\{v_{i-1}(w),v_i(w),v_{\lambda_w}(w)\}\cup K_i\cup L_i$ for $i\in\{2,\dots,\lambda_w\}$.
    Observe that all $d^w_i$ are leaves of $T$, and any terminal that they contain is also present in $c^w_i$ for all $i\in[\lambda_w]$.
    Therefore, each such bag is of type~\ref{type2} and the decomposition $(T,\mathcal{B})$ is of the desired form.

    If $G$ is not connected, we consider the connected components $G_1,\dots,G_{\gamma}$ of $G$ for some $\gamma\in\N$.
    By the previous arguments, we obtain a tree-decomposition $(T_i,\mathcal{B}_i)$ for each $i\in[\gamma]$.
    We combine the tree-decompositions into one central decomposition in the following way.

    We construct the tree $T$ as a $\gamma$-comb on the vertex set $\{c_1,\dots,c_\gamma,d_1,\dots,d_\gamma\}$, where $d_i$ is replaced by $T_i$, such that $c_i$ attaches to the central vertex of $T_i$.
    This increases the maximum degree of the tree-decomposition by at most one.
    The collection of bags $\mathcal{B}$ is obtained by copying the $\mathcal{B}_i$ for $i\in[\gamma]$.
    Note that in this way, the bags corresponding to $c_i$ for $i\in[\gamma]$ remain empty.
    \qed
\end{proof}

\section{Dynamic Program}\label{sec:dp}
In this section, we describe how to exploit the structural results from \Cref{sec:subdet,sec:struct} into an efficient algorithm for Problem~\ref{prob_reduced}.
We propose a \emph{bottom-up} dynamic program on the tree-decomposition given by \Cref{thm_refinedstruct}, which follows the ideas of the dynamic program presented in~\cite{AFJ24}.
We define a set of \emph{rooted instances} corresponding to subtrees of the tree-decomposition.
In each of these subtrees, we compute optimal solutions for the possible assignments on the right hand sides of the extra rows as well as the possible assignments on variables in the intersection with the parent in the tree-decomposition.

In bags of type~\ref{type2}, we fix all assignments on the terminals of the subgraph, and solve the remaining integer program with at most two non-zeros per row.
Bags of type~\ref{type1} have bounded size. 
Hence, we can enumerate all possible assignments to the corresponding variables.
Since the degree of the tree-decomposition is bounded, we can also combine the optimal solutions in the corresponding subtrees efficiently.
Note that by \Cref{thm_cookproximity} and \Cref{thm_refinedstruct}, the number of possible assignments is polynomial in $n$ and $\Delta$.

In the context of the dynamic program, we work with \emph{rooted tree-decompositions}.
A rooted tree-decomposition of a graph $G$ is a tree-decomposition $(T,\mathcal{B})$ with the additional property that some vertex $t_0\in V(T)$ is identified as the root of $T$.
In a rooted tree-decomposition $(T,\mathcal{B})$ with root $t_0$, any vertex $t$ of $T$ induces a \emph{rooted subtree}, i.e., the tree induced on the vertices $t'\in V(T)$, such that the unique $t'-t_0$-path in $T$ contains $t$.
We denote this rooted subtree by $T(t)$, and its associated collection of vertices of $G$ by $B(t):=\{v\in V(G)\mid\exists t'\in V(T(t)):v\in B_{t'}\}$.
Further, for any vertex of a rooted tree-decomposition $t\in V(T)$, we define its \emph{adhesion} as $\mathrm{adh}(t):=\emptyset$ if $t=t_0$, and $\mathrm{adh}(t):=B_t\cap B_{t'}$ for $t'$ being the unique parent of $t$ in the rooted tree-decomposition otherwise.

We begin by recalling the setup for Problem~\ref{prob_reduced}.
Let $\Delta,k\in\N$, $A\in\{0,\pm1\}^{m\times n}$, $W\in\Z^{k\times n}$, such that $A$ has two non-zeros per row.
Let $M:=\smallmat{A\\W}$ be totally $\Delta$-modular, $c,\ell,u\in\Z^n$, $b\in\Z^{m}$, and $d\in\Z^k$, which together define an instance of Problem~\ref{prob_reduced}.
Let $S=(V,E,\sigma)$ be a signed graph with corresponding incidence matrix $A$, and denote $G:=(V,E)$ the unsigned graph with set of terminals $R\subseteq V$, where $v\in V$ is a terminal if the column vector of $W$ corresponding to $v$ is non-zero.
By \Cref{prop_alt} and \Cref{lemma_decorated}, we conclude that $G$ has no $(2k\Delta+1)$-decorated tree as a subgraph.

Let $(T,\mathcal{B})$ be the tree-decomposition of $G$ given by \Cref{thm_refinedstruct} with respect to the set of terminals $R$.
We turn $(T,\mathcal{B})$ into a rooted tree-decomposition by identifying the root as $c_1$ (in the notation of the proof of \Cref{thm_refinedstruct}), i.e., the first vertex on the path of the comb that combines the decompositions of the components.

Further, we define assignment spaces for the right hand side and the variables respectively: $\mathcal{D}:=[-n\Delta f_{\ref{thm_cookproximity}}(n,\Delta),n\Delta f_{\ref{thm_cookproximity}}(n,\Delta)]^k\cap\Z^k$, and $\mathcal{X}_t:=[-f_{\ref{thm_cookproximity}}(n,\Delta),f_{\ref{thm_cookproximity}}(n,\Delta)]^{\mathrm{adh}(t)}\cap\Z^{\mathrm{adh}(t)}$.

\begin{definition}
    Given a vertex of the tree-decomposition $t\in V(T)$, an assignment on the adhesion to the parent $\chi_t\in\mathcal{X}_t$, and a target for the extra constraints $d_t\in\mathcal{D}$, we define a rooted instance of Problem~\ref{prob_reduced}.
    We denote by $A_t$ the restriction of $A$ to the columns corresponding to vertices of $B(t)$ and to the rows corresponding to edges induced on $G[B(t)]$.
    Similarly, we denote by $W_t,c_t,\ell_t$, and $u_t$ the restriction of $W,c,\ell$ and $u$ to the vertices of $B(t)$, where $W_t(i,v)=0$ and $c_t(i,v)=0$ for $v\in\mathrm{adh}(t)$ and $i\in[k]$.
    Finally, we denote by $b_t$ the restriction of $b$ to the edges induced on $G[B(t)]$.

    The rooted problem is defined as 
    \[
        \min \, \left\{ c_t^\intercal x : A_tx \le b_t, \, W_tx=d_t, \, \ell_t\le x\le u_t,\, x(v)=\chi_t(v) \forall v\in\mathrm{adh}(t), \, x \in \Z^{B(t)} \right\}.
    \]
    We denote an instance of the rooted problem by $I:=I(t,\chi_t,d_t)$, and its optimal value by $\mathrm{OPT}(I)$.
    By restricting to the vertices and edges in a single bag $B_t$ instead of a rooted subtree $B(t)$, we define a \emph{local} instance $I^0(t,\chi_t,d_t)$ with optimal value $F^0_t(\chi_t,d_t)$.
    We denote any optimal solution of $I^0(t,\chi_t,d_t)$ by $x^0_t(\chi_t,d_t)$.
\end{definition}

Note that $I(t_0,\emptyset,d)$ is the original instance of Problem~\ref{prob_reduced}.
Thus, in order to solve Problem~\ref{prob_reduced} it is sufficient to give a recursive algorithm on how to solve rooted problems.
Further, for any $t\in V(T)$, there is $(2n\Delta f_{\ref{thm_cookproximity}}(n,\Delta)+1)^k\cdot(2f_{\ref{thm_cookproximity}}(n,\Delta)+1)^{O((k\Delta)^2)}$ different possible assignments for rooted instances, which is polynomial for fixed $\Delta,k\in\N$.

\begin{definition}\label{def_recursive}
    Let $I(t,\chi_t,d_t)$ be a rooted instance corresponding to vertex $t\in V(T)$.
    Let $t_1,\dots,t_q$ denote the children of $t$ in $T$ (where $q=0$  if $t$ is a leaf).
    We define
    \begin{align*}
        F_t(\chi_t,d_t):=\min\bigg\{F^0_t(\chi_t,&d^0_t)+\sum_{i\in[q]}(F_{t_i}(\chi_{t_i},d_{t_i}))\,:\,\chi_{t_i}\in\mathcal{X}_{t_i}\forall i\in[q],\,\chi_t(v)=\chi_{t_i}(v)\\
        &\forall i\in[q],\,v\in B_t\cap B_{t_i},\,d^0_t,d_{t_1},\dots,d_{t_q}\in\mathcal{D},\,d^0_t+\sum_{i\in[q]}d_{t_i}=d_t\bigg\}.
    \end{align*}

    We denote by $x_t(\chi_t,d_t)$ the composition of $x^0_t(\chi_t,d_t)$ with $x_{t_i}(\chi_{t_i},d_{t_i})$ for $i\in[q]$.
\end{definition}

\begin{proof}[Proof of \Cref{thm_main}]
    By \Cref{prop_reduction}, it suffices to give a polynomial-time algorithm for instances of Problem~\ref{prob_reduced}.
    By \Cref{thm_refinedstruct} we obtain a tree-decomposition $(T,\mathcal{B})$.
    We show that $F_t(\chi_t,d_t)=\mathrm{OPT}(I(t,\chi_t,d_t))$ for all $t\in T$, all $\chi_t\in\mathcal{X}_t$, and $d_t\in\mathcal{D}$.
    We proceed by induction on the distance of a vertex $t\in T$ from the root $t_0$, starting with the largest distance.
    
    The claimed algorithm corresponds to computing $F_{t_0}(\emptyset,d)$ as in \Cref{def_recursive}.

    By definition, the statement is true for any $t\in V(T)-t_0$, where $t$ is a leaf.
    Assume that $t$ is not a leaf.
    Then, $t$ has $q$ children $t_1,\dots,t_q$, such that for each $i\in[q]$, and each $\chi_{t_i}\in\mathcal{X}_{t_i}$ and $d_{t_i}\in\mathcal{D}$, it holds by induction that $F_{t_i}(\chi_{t_i},d_{t_i})=\mathrm{OPT}(I(t_i,\chi_{t_i},d_{t_i}))$.
    Recall from \Cref{sec:reduction}, that $\ell,u\in[-f_{\ref{thm_cookproximity}}(n,\Delta),f_{\ref{thm_cookproximity}}(n,\Delta)]^n\cap\Z^n$.
    Consider an optimal solution $x_t^*(\chi_t,d_t)\in[-f_{\ref{thm_cookproximity}}(n,\Delta),f_{\ref{thm_cookproximity}}(n,\Delta)]^n\cap\Z^n$ of $I(t,\chi_t,d_t)$.
    For $i\in[q]$, we can decompose $x_t^*(\chi_t,d_t)$ into solutions for $I_{t_i}(\chi_{t_i},d_{t_i})$, where $\chi_{t_i}$ and $d_{t_i}$ are induced by the corresponding assignments of $x_t^*(\chi_t,d_t)$.
    By definition of the spaces $\mathcal{X}_{t_i}$ for $i\in[q]$ and $\mathcal{D}$, we can see that the corresponding optimal values have been pre-computed.
    Further, the dynamic program explores all valid assignments for variables in $B_t$.
    Thus, $c_t^\intercal x_t^*(\chi_t,d_t)=c_t^\intercal x_t(\chi_t,d_t)$, certifying optimality.

    We conclude the proof with an analysis of the runtime.
    The total size of the tree-decomposition is linear in $n$, i.e., $|V(T)|\in O(n)$.
    For each vertex $t\in V(T)$, we obtain at most $(2n\Delta f_{\ref{thm_cookproximity}}(n,\Delta)+1)^k\cdot(2f_{\ref{thm_cookproximity}}(n,\Delta)+1)^{O((k\Delta)^2)}$ relevant rooted instances.
    Each local instance can be solved in polynomial time.
    For bags of type~\ref{type1}, this corresponds to enumerate all possible assignments within the proximity bound.
    For bags of type~\ref{type2}, this corresponds to solving an integer program with bounded subdeterminants and two non-zeros per row, see \Cref{thm_twononzeros}.
    Solving a rooted instance corresponds to solving a local instance and combining it with a polynomial number of rooted instances, which can again be done in polynomial time.
\end{proof}
\begin{proof}[Proof of \Cref{cor_tu}]
    The statement follows from \Cref{thm_main} and its proof with the following additional observations.

    Since $A$ is totally unimodular, we can solve the local problem for bags of type~\ref{type2} such that the exponent of the running time does not depend on $\Delta$, see e.g. Tardos~\cite{Tar86}.
    Further, here the proximity function $f_{\ref{thm_cookproximity}}(n,\Delta)$ does not depend on the dimension of the problem $n$, but just on $\Delta$ and $k$, see Aprile et al.~\cite[Theorem~4]{AFJ24}.
\end{proof}

\section{Conclusion}\label{sec:conclusion}
We study IPs with bounded subdeterminants and give a strongly polynomial-time algorithm to solve a new structured class of IPs.
We achieve our results by combining structural results with a dynamic programming approach.
Within the dynamic program, we use the algorithm to solve integer programs with bounded subdeterminants and two non-zeros per row~\cite{FJWY25} in a blackbox fashion.
In particular, for our main result, the proximity result by Cook et al.~\cite{CGST86} can be used directly.
This is in contrast to other recent papers, that make use of auxiliary proximity or solution decomposition results, see e.g.~\cite{AWZ17,EW19,ERW24,AFJ24,FJWY25}.
Still, this leaves open, whether one can prove a stronger proximity result for Problem~\ref{prob_reduced} or its generalizations.
Such a proximity result could open the door to more general results on the solvability of totally $\Delta$-modular IPs.

In the work by Fiorini et al.~\cite{FJWY25}, the majority of the technical difficulty lies in constraint matrices with bounded subdeterminants and two non-zeros per \emph{row}. 
They also solve the case with two non-zeros per \emph{column}.
We do not extend this part of their result to additional rows and columns.
Note that such integer programs can model a version of the \emph{exact matching} problem~\cite{PY82} with a structural restriction.
In the exact matching problem, we are given a graph with red and blue edges, and $k\in\N$.
The goal is to decide whether there exists a perfect matching with exactly $k$ red edges.
This problem is not known to be $\mathrm{NP}$-hard, but no efficient deterministic algorithm is known.
Subproblems involving the exact matching problem have proven to be complex to solve deterministically in the literature on subdeterminants, see~\cite{NSZ24,NNSZ24,AFJ24}.

Finally, we remark on a structural question for matrices with bounded subdeterminants.
In many recent papers, it has been helpful to exploit a relation between such matrices and certain graph classes, see~\cite{AWZ17,FJWY25,NSZ24,NNSZ24,AFJ24}.
We ask whether this connection can be made more precise, and if the question of solving integer programs with bounded subdeterminants can be solved by analyzing the right graph classes.
One indication for such a relation can be found in recent work on the number of disjoint columns in matrices with bounded subdeterminants, see~\cite{LPSX22,AS23}.

\section*{Acknowledgements}
The author thanks Manuel Aprile, Samuel Fiorini, Caleb McFarland, Stefan Weltge and Yelena Yuditsky for helpful discussions.
The author acknowledges funding by the ERC Consolidator Grant 615640-ForEFront and the \emph{Fonds de la Recherche Scientifique} - FNRS through research project BD-DELTA (PDR 20222190, 2021–24).

In addition, we would like to thank the anonymous reviewers for their careful reading and their helpful comments which greatly helped to improve the paper.

\bibliographystyle{splncs04}
\bibliography{references}

\end{document}